\documentclass[runningheads]{llncs}
\usepackage[T1]{fontenc}
\usepackage{xcolor}
\usepackage{graphicx}
\usepackage{xspace}
\usepackage{amsmath,amssymb}

\newcommand{\mw}{mw}
\newcommand{\cext}{\chi_{\text{1-ext}}}
\newcommand{\newsubsetsum}{\textsc{Generating Set}\xspace}

\usepackage{tikz}
\usetikzlibrary{shapes}
\usepackage{pgfplots}
\usepgfplotslibrary{groupplots}

\usepackage{comment}
\definecolor{lipicsYellow}{RGB}{255,211,67}
\definecolor{lipicsLightGray}{RGB}{200,200,200}
\definecolor{lipicsGray}{RGB}{128,128,128}

\usepackage[caption = false]{subfig}

\usepackage{amsthm}

\begin{document}

\title{Channel allocation revisited through 1-extendability of graphs}
%
%
\author{Anthony Busson\and
Malory Marin\and
Rémi Watrigant}
\authorrunning{A. Busson et al.}
%
\institute{Univ Lyon, CNRS, ENS de Lyon, Université Claude Bernard Lyon 1, LIP UMR5668, France}

\maketitle

\begin{abstract}
We revisit the classical problem of channel allocation for Wi-Fi access points (AP). Using mechanisms such as the CSMA/CA protocol, Wi-Fi access points which are in conflict within a same channel are still able to communicate to terminals. In graph theoretical terms, it means that it is not mandatory for the channel allocation to correspond to a proper coloring of the conflict graph. However, recent studies suggest that the structure--rather than the number--of conflicts plays a crucial role in the performance of each AP. More precisely, the graph induced by each channel must satisfy the so-called $1$-extendability property, which requires each vertex to be contained in an independent set of maximum cardinality. 
In this paper we introduce the $1$-extendable chromatic number, which is the minimum size of a partition of the vertex set of a graph such that each part induces a $1$-extendable graph. We study this parameter and the related optimization problem through different perspectives: algorithms and complexity, structure, and extremal properties.
We first show how to compute this number using modular decompositions of graphs, and analyze the running time with respect to the modular width of the input graph. We also focus on the special case of cographs, and prove that the $1$-extendable chromatic number can be computed in quasi-polynomial time in this class.
Concerning extremal results, we show that the $1$-extendable chromatic number of a graph with $n$ vertices is at most $2\sqrt{n}$, whereas the classical chromatic number can be as large as $n$. We are also able to construct graphs whose $1$-extendable chromatic number is at least logarithmic in the number of vertices.
\end{abstract}

\section{Introduction}

\subsection{CSMA/CA network and channel }

\paragraph{Wi-Fi networks.} Wi-Fi networks are one of the primary means of accessing the Internet today. A Wi-Fi network is composed of one or more access points (AP) giving Internet access to stations associated to them. Each access point uses exclusively a particular channel, which consists in a 20 MHz wide frequency band. There are currently 3 independent channels (with no common frequencies) in the 2.4GHz band and around 23 in the 5GHz band. These channels can be aggregated, enabling an AP to use a 40Mhz frequency range, thus improving throughput. 
To arbitrate transmissions between access points and stations using the same channel, the CSMA/CA protocol is used. This is a listen-before-talk mechanism that enables time-sharing of the channel between stations and APs. Time-sharing is not static/predetermined, but opportunistic.  For example, when two APs use the same channel, the number of transmissions is on average shared equally between the two APs (in an equivalent radio environment). The throughput is therefore approximately halved for each of the access points, compared with the case where they are alone on their channel. Such sharing is referred to as conflicts in the literature, and is represented by a conflict graph where the vertices are the APs and the edges indicate whether they share the channel. We consider only the DCF mode of Wi-Fi, which is the mode used in practice on commercial access points.

\paragraph{Channel allocation.} Channel allocation is the algorithm used to set channels on the various access points. There is then a conflict graph for each channel. The aim of the algorithm might be to minimize the number of edges in the conflict graphs, so as to have as little sharing as possible, or to distribute the network load over the different channels. The algorithms implemented on commercial Wi-Fi networks are variants of these two approaches. However, these approaches perform poorly when traffic on the access points become saturated (i.e. when an AP always has a frame to transmit).  Depending on the topology of the conflict graph, some APs may starve and be unable to transmit any frames at all. Indeed, the performance of each AP depends on the structure of its conflict graph. More precisely, it appears that the structure of its \textit{independent sets} plays a crucial role, where an independent set of a graph is a set of pairwise non-adjacent vertices.

\paragraph{A graph theoretical approach.}
The key metric describing the performance of a channel allocation is the proportion of accesses (number of transmissions) that each vertex of the conflict graph obtains. It is denoted by $p_v$ for a vertex $v$. Hence, $p_v=1$ if $v$ has no neighbor in the conflict graph, and $0 \leqslant p_v < 1$ otherwise. If $p_v=0$ the vertex is in starvation and cannot transmit any frame.
The first formal study characterizing $p_v$ was presented in \cite{Liew2007BackoftheEnvelopeCO}, where it was demonstrated that under saturation conditions, $p_v$ can be calculated as:
$$
p_v = \frac{\sum_{S \in \mathcal{S}(G) : v\in S}\theta^{|S|}}{\sum_{S \in \mathcal{S}(G)}\theta^{|S|}}
$$
where $\theta$ is the ratio between transmission and listen phase durations, and $\mathcal{S}(G)$ is the set of independent sets of $G$. When $\theta$ tends to infinity, $p_v$ tends to the ratio of the number of independent sets of maximum cardinality (maximum independent set, or MIS for short) that contains $v$ ($\#_v \alpha(G)$) over the total number of MIS of $G$($\# \alpha(G)$) :
$$
\lim_{\theta\rightarrow + \infty}  p_v = \frac{\#_v \alpha(G)}{\# \alpha(G)}
$$ 
To ensure efficient channel utilization, $\theta$ is often set to a large value in practice. For example, in Wi-Fi networks, typical values of $\theta$ range from 20 to 100, depending on the transmission parameters. However, with such values, a node or vertex that does not belong to any MIS will experience a very low throughput. Figure \ref{figureIntro} illustrates this phenomenon on two simple graphs.

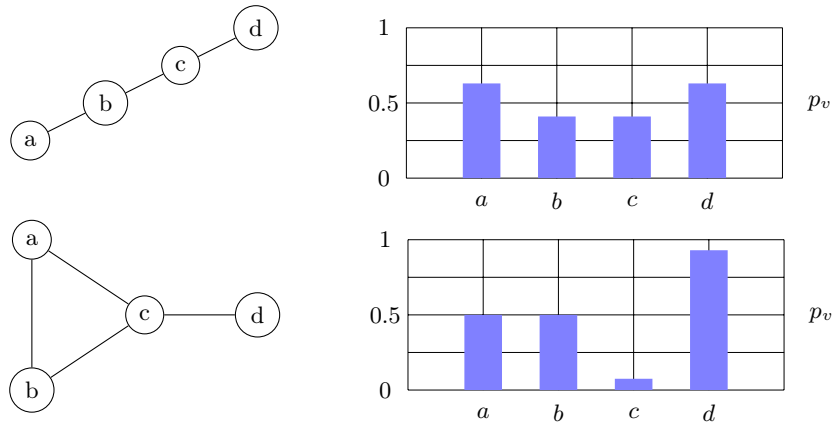
\begin{figure}[t!]
\centering
\begin{minipage}{1\linewidth}
\centering
\begin{tikzpicture}
\node[draw, circle] (a) at (0,0.5) {a};
\node[draw, circle] (b) at (1,1) {b};
\node[draw, circle] (c) at (2,1.5) {c};
\node[draw, circle] (d) at (3,2) {d};

\draw (a) -- (b) -- (c) -- (d) ;

\draw[very thin] (5,0) grid[ystep=0.5] (10,2);
\draw[ycomb, line width=5mm,color=blue!50] plot coordinates {(6,2*0.63) (7,2*0.41) (8,2*0.41) (9,2*0.63)} ;

\foreach \x in {0,0.5,1}{
	\node[] () at (4.7, \x*2) {\x};
}

\node[] () at (10.5, 1) {$p_v$};

\node[] () at (6,-0.3) {$a$};
\node[] () at (7,-0.3) {$b$};
\node[] () at (8,-0.3) {$c$};
\node[] () at (9,-0.3) {$d$};

\end{tikzpicture}
\end{minipage}
\begin{minipage}{1\linewidth}
\centering
\begin{tikzpicture}
\node[draw, circle] (a) at (0,2) {a};
\node[draw, circle] (b) at (0,0) {b};
\node[draw, circle] (c) at (1.5,1) {c};
\node[draw, circle] (d) at (3,1) {d};

\draw (a) -- (b) -- (c) -- (d) ;
\draw (a) -- (c) ;

\draw[very thin] (5,0) grid[ystep=0.5] (10,2);
\draw[ycomb, line width=5mm,color=blue!50] plot coordinates {(6,2*0.5) (7,2*0.5) (8,2*0.075) (9,2*0.93)} ;

\foreach \x in {0,0.5,1}{
	\node[] () at (4.7, \x*2) {\x};
}

\node[] () at (10.5, 1) {$p_v$};

\node[] () at (6,-0.3) {$a$};
\node[] () at (7,-0.3) {$b$};
\node[] () at (8,-0.3) {$c$};
\node[] () at (9,-0.3) {$d$};

\end{tikzpicture}
\end{minipage}
\caption{Two graphs and the $p_v$ value of each vertex, obtained with ns-3 simulations. ns-3 is a network simulator that implements the Wi-Fi protocol in the same way as the standards, and incorporates a realistic radio layer. The simulated Wi-Fi corresponds to the IEEE 802.11ax amendment with frame aggregation enabled and a physical transmission rate of 144.4Mbit/s.  The top graph is $1$-extendable and thus each vertex has a non-negligible throughput. The bottom graph is not $1$-extendable, and the middle vertex has a very low throughput.}\label{figureIntro}
\end{figure}

To guarantee that every vertex has a fair share of the channel and to prevent any node from being deprived of transmission opportunities, it is essential that every vertex belongs to an MIS. In graph theoretical terms, each channel must induce a conflict graph that is \textit{$1$-extendable}.
In the Wi-Fi context, this means that channel allocation must lead to 1-extendable conflict graphs. 
In this article, we investigate from a structural and algorithmic point of view whether such allocations are possible with a fixed number of channels.

\subsection{Definitions and Related Work}


A graph $G$ is \emph{$1$-extendable} if each vertex belongs to at least one maximum independent set (MIS). This class was initially introduced by Berge \cite{BERGE198231} under the name of B-graphs in order to study \textit{well-covered} graphs. The class of well-covered graphs was introduced by Plummer in 1970 \cite{PLUMMER197091}, and is defined to be the graphs where  every independent set that cannot be further extended is also an MIS. In other words, all (inclusion-wise) maximal independent sets have the same size. 
This property makes them particularly interesting as it guarantees that the greedy algorithm for constructing a large independent set always produces an optimal solution. In particular, a graph is well-covered if, and only if, for any $k\geqslant 1$ any independent set of size $k$ is a subset of an MIS. This last property, when $k$ is fixed, was called $k$-extendability by Dean and Zito \cite{DEAN199467}. 
More recently, 
the computational aspects of testing $1$-extendability has been studied \cite{Berge}. Specifically, it has been shown that recognizing $1$-extendable graphs is NP-hard, even for unit disk graphs, which are a common model for wireless networks.
A more involved computational complexity of recognizing $1$-extendable graphs, well-covered graphs and generalizations was also conducted in ~\cite{FeMaWa24}, where in particular it was proved that deciding whether a given graph is $1$-extendable is $\Theta_2^p$-complete, and thus belong to a complexity class beyond $NP$.\newline

For a graph $G$, we say that a partition $V_1$, $\dots$, $V_k$ of its vertex set is a \textit{1-extendable $k$-partition} if the subgraph induced by $V_i$ is $1$-extendable for every $i = 1, \dots, k$. The \textit{1-extendable chromatic number} of a graph $G$, denoted by $\cext(G)$, is the smallest integer $k$ such that $G$ admits a $1$-extendable $k$-partition. Observe that this number is bounded above by the ``classical'' chromatic number $\chi(G)$, which is the minimum integer $k$ such that $V(G)$ has a partition into $k$ independent sets (since an independent set is $1$-extendable). The main objective of this paper is to study this new parameter both structurally and algorithmically. Given a graph $G$ and an integer $k$, the \textsc{1-Extendable Partition} problem asks whether $\cext(G) \leqslant k$. In a nutshell, if $G$ represents the conflict graph of a set of Wi-Fi access points, $\cext(G)$ denotes the minimum number of channels in an assignment which guarantees that under saturation, no access point will be in starvation.\newline

One classical way of attacking a graph optimization problem is to decompose the input into simpler pieces in a way that rearranging the pieces behaves well with the considered problem. As the $1$-extendability deals with the structure of independent sets, it makes sense to consider a decomposition which preserves this structure. 
One such object is the so-called \textit{modular decomposition}, which is a recursive representation of a graph in which each piece--a subset of vertices called \textit{module}-- has a homogeneous behavior with respect to the others (a formal definition will be given in Section~\ref{sec:preliminaries}). It is important to note that every graph admits a modular decomposition, and it can be obtained in linear time~\cite{McCoSpi99}. Moreover, it is worth noting that modular decompositions were also considered for studying well-covered graphs~\cite{MiPi23}. When dealing with the modular decomposition of a graph, a natural parameter measuring its complexity is the so-called \textit{modular-width}, which informally represents the arity of the branching tree of the recursive decomposition. For a graph $G$, we denote its modular-width by $mw(G)$. 


\subsection{Contributions and organization of the paper}

In this paper we conduct an analysis of the $1$-extendable chromatic number through several different angles: complexity, algorithms, structure and extremal behavior.

In Section~\ref{sec:complexity}, we first study the computational complexity of deciding the exact value of this number for arbitrary graphs. It was recently shown that the problem of deciding whether a graph is $1$-extendable belongs to a complexity class above $NP$~\cite{FeMaWa24}. Namely, the problem is complete for $\Theta_2^p$, which is the set of problems that can be solved in polynomial time with a logarithmic number of calls to a \textsc{Sat} oracle. It is natural to expect that computing the $1$-extendable chromatic number is not easier by proving that \textsc{1-Extendable k-Partition} is $\Theta_2^p$-hard for every fixed $k\geqslant 1$.

We then turn to positive results. 
We first obtain an algorithm for deciding whether a graph is $1$-extendable which runs in single exponential in the modular-width of the input graph, and linear time in its number of vertices. We then use this algorithm to compute the $1$-extendable chromatic number is at most $k$ in time $\alpha(G)^{O(mw(G)k)}\cdot n$, where $\alpha(G)$ is the \textit{independence number} of a graph $\alpha(G)$, which corresponds to the size of an MIS of $G$.

A building block of these algorithms is the special case of \textit{cographs}, which can be seen as the simplest graphs with respect to modular decomposition. In particular, cographs are exactly the graphs with modular-width $2$, the smallest possible value. We are able to obtain an algorithm for deciding whether the $1$-extendable chromatic number of a cograph $G$ is at most $k$ in time $k \cdot \alpha(G)^{O(k)}$.
%
We also exhibit a number-theoretical formulation of the problem in the case of complete multipartite graphs, which is an even more restricted case of cographs, but seems to contain the main combinatorial complexity of the problem.

In Section~\ref{sec:extremal}, we investigate the extremal behavior of the $1$-extendable chromatic number. We first show that this number is at most $2\sqrt{n}$ in general graphs with $n$ vertices.

For the special case of cographs, we are even able to obtain a logarithmic bound, which appears to be tight, as there exist a cograph which requires a logarithmic bound. This result, together with the previously mentioned algorithm, implies an algorithm for deciding the $1$-extendable chromatic number of a cograph $G$ running in quasi-polynomial time, namely $k \cdot \alpha(G)^{O(\log_2(\alpha(G)))}$.

Finally, we conclude the paper in Section~\ref{sec:conclusion} with several open questions.\newline

\section{Preliminaries}\label{sec:preliminaries}

Most definitions used in this paper are standard. We nevertheless recall those which are heavily used. For a non-negative integer $n$, let $[n] = \{1, \dots n\}$. Given a graph $G$, $V(G)$ denotes its vertex set and $E(G)$ its edge set. Given $R\subseteq V(G)$, we use $G[R]$ to denote the subgraph of $G$ induced by $R$, and $G-R = G[V(G)-R]$ to denote the subgraph of $G$ induced by the complement of $R$.
Unless otherwise stated, $n$ always denotes the number of vertices of the considered graph.
The \textit{disjoint union} of two graphs $G_1$ and $G_2$ is the graph denoted by $G_1 \cup G_2$ whose vertex set is $V(G_1) \cup V(G_2)$ and edge set is $E(G_1) \cup E(G_2)$. 
The \textit{complete sum} of two graphs $G_1$ and $G_2$ is the graph denoted by $G_1 + G_2$ whose vertex set is $V(G_1) \cup V(G_2)$ and edge set is $E(G_1) \cup E(G_2) \cup \{\{x, y\}: x \in V(G_1), y \in V(G_2)\}$.
Given a graph $H$ with $V(H) = \{v_1, \dots, v_n\}$ and $n$ graphs $G_1$, $\dots$, $G_n$, the substitution operation of $(G_1, \dots, G_n)$ on $H$ is the graph whose vertex set is $\cup_{i=1}^n V(G_i)$, and edge set $\left( \cup_{i=1}^n E(G_i) \right) \cup \left( \cup_{\{v_i, v_j\} \in E(H)} \{\{x, y\}, x \in V(G_i), y \in V(G_j)\} \right)$.
We say that $H$ is the \textit{representative graph} of this new graph.
Notice that the disjoint union (resp. complete sum) operation corresponds to the substitution operation where $H$ is the graph composed of two non-adjacent (resp. adjacent) vertices.
Any graph can be recursively constructed from the graph with one vertex using only substitution operations, and the \textit{width} of such a decomposition is the maximum order of the graphs $H$ used along the process. The \textit{modular-width} of a graph $G$, denoted by $\mw(G)$ is the minimum width of such a decomposition. It is known that a decomposition of minimum width is unique. It  is called the \textit{modular decomposition} of $G$, and can be found in linear time~\cite{Habib}. 
The modular decomposition can naturally be represented as a rooted tree, where each leaf represents the introduction of a new single vertex, and each internal node is labeled by the graph $H$ on which the substitution is done.
Cographs are exactly the graphs of modular-width at most two, that is, graphs recursively obtained from a single vertex using only disjoint unions and complete sums. The corresponding tree representation of a cograph is called its \textit{cotree}. 
A \textit{module} of a graph $G$ is a set of vertices $M$ such that every vertex $v \in V(G) \setminus M$ is either adjacent to all vertices of $M$ or not adjacent to any vertex of $M$, and it is a \textit{strong module} if for any other module $M' \neq M$ the intersection $M \cap M'$ is either empty or equals $M$ or $M'$. Observe that in the substitution of $(G_1, \dots, G_n)$ on $H$, $V(G_i)$ forms a module in the new graph. In a modular decomposition, such modules are always strong modules.
A graph  $G$ is a \textit{complete multipartite graph} if there exists a partition of $V(G)$ into subsets $V_1, \cdots ,V_q$ such that $\{x, y\}$ is an edge if and only if $x \in V_i$ and $x \in V_j$ with $i \neq j$. Observe that a complete multipartite graph is a cograph. A complete multipartite graph is called \textit{balanced} if $|V_1|=...=|V_q| = \ell$, and we call $\ell$ the \textit{width} of $G$ and $q$ its \textit{length}.
A \textit{Fixed Parameter Tractable} (FPT) algorithm is an algorithm deciding whether an instance of a decision problem is positive in time $f(k)\cdot n^{c}$, where $c$ is some fixed constant, $f$ is a computable function, $n$ is the size of the instance, and $k$ is some chosen parameter of it. If $c=1$ (resp. $c=2$) we say this is a \textit{linear} (resp. \textit{quadratic}) FPT algorithm For more details, see Cygan \textit{et al.} \cite{Cygan}.

\section{Complexity and algorithms} \label{sec:complexity}

Deciding if a graph is $1$-extendable was shown to be NP-hard in \cite{Berge}, even when restricted to subcubic planar graphs and unit disk graphs. Note that the problem is not in NP; it is, in fact, complete for the class $\Theta_2^p = P^{NP[\log]}$, which is the set of problems that can be solved in polynomial time with a logarithmic number of calls to a \textsc{Sat} oracle. This class is a subset of $\Sigma_2^p = P^{NP}$ in the polynomial hierarchy. We exhibit a reduction from testing $1$-extendability to \textsc{1-Extendability Partition}, and as a corollary, show that the problem is $\Theta_2^p$-hard, which is stronger than simple NP-hardness.

\begin{theorem}\label{thm:hardness}
\textsc{1-Extendable $k$-Partition} is $\Theta_2^p$-hard for every~fixed $k~\geqslant~1$.
\end{theorem}

\begin{proof}
We prove the result by induction on $k$. The result is true for $k=1$, as mentioned previously. Let $k > 1$, and assume that the problem of finding a 1-extendable $(k-1)$-partition is NP-hard. Let $G$ be a graph on $n$ vertices. We define $G' = G+S$ where
$S$ is a graph composed of a unique independent set of $kn+1$ new vertices. We call $V'$ the set of vertices of $G'$. Since $G'$ has $(k+1)n+1$ vertices, its size is polynomial in the size of $G$.

If $\cext(G)\leqslant k-1$, then one can construct a $1$-extendable $k$-partition of $G'$ by using a whole new color for $S$. Thus, $\cext(G') \leqslant k$.

Conversely, if $\cext(G') \leqslant k$, let $(V_i')_{1\leqslant i \leqslant k}$ be a partition such that $G'[V_i']$ is 1-extendable for any $1\leqslant i\leqslant k$. Let $V_i = V\cap V_i'$ and $I_i = V(S)\cap V_i'$  for all $1\leqslant i \leqslant k$. The $I_i$'s form a partition of $V(S)$ into $k$ sets, and we assume without loss of generality that $I_k$ has maximum cardinality. We have 
$$
|I_k| \geqslant \frac{|S|}{k} > n
$$
Now, assume that $V_k \neq \emptyset$. Since $G'[V_k']$ is 1-extendable, $G[V_k]+S[I_k]$ is 1-extendable and thus $\alpha(G[V_k]) = \alpha(S[I_k])> n$. However, $\alpha(G[V_k])\leqslant n$ since $G[V_k]$ is a subgraph of $G$, which leads to a contradiction. So $V_k = \emptyset$, and thus the color $k$ is not used on $G$.  
Since any $G[V_i]$ is 1-extendable for $1\leqslant i \leqslant k-1$, we obtain a 1-extendable $(k-1)$-partition of $G$.
\end{proof}

However, note that \textsc{1-Extendable Partition} does not appear to be in $\Theta_2^p$ due to the additional complexity of the partition problem. The lowest class in the polynomial hierarchy where the problem seems to lie is $\Delta_2^p = \text{NP}^{\text{NP}}$. It remains an open question whether the problem is complete for this class.

\subsection{Cographs}

We start with the simplest graphs in terms of modular decomposition: cographs. Due to the structure of their cliques and independent sets, cographs are a natural candidate for being easier instances of problems related to the notion of 1-extendability. Indeed, the recursive definition of cographs leads to a linear-time algorithm for determining the independence number. Similarly, we can also characterize 1-extendable cographs.

\begin{proposition}\label{OpCographs}
Let $G_1$ and $G_2$ be two cographs. 
\begin{itemize}
\item  $G_1\cup G_2$ is 1-extendable if, and only if, $G_1$ and $G_2$ are 1-extendable.
\item $G_1+ G_2$ is 1-extendable if, and only if, $G_1$ and $G_2$ are 1-extendable and $\alpha(G_1)=\alpha(G_2)$.
\end{itemize}
\end{proposition}

\begin{proof}
First, given two graphs $G_1=(V_1,E_1)$ and $G_2=(V_2,E_2)$, note that $\alpha(G_1\cup G_2)=\alpha(G_1)+\alpha(G_2)$.

Assume that $G_1$ and $G_2$ are $1$-extendable, and let $v\in V$. Whithout loss of generality, we assume $v\in V_1$ and let $S_1$ be a maximum independent set of $G_1$ that contains $v$. In addition, let $S_2$ be any maximum independent set of $G_2$. Then, $S_1\cup S_2$ is an independent set of cardinality $\alpha_1+\alpha_2$ and thus $v$ is in an MIS of $G$.

Assume that $G_1$ or $G_2$ is not $1$-extendable. Whithout loss of generality, we assume that $G_1$ is not and let $v\in V_1$ be a vertex which is not in an MIS of $G_1$. By contradiction, if $G$ is $1$-extendable, then $v$ is in an MIS $S$ of $G$. Let $S_1=S\cap V_1$ and $S_2=S\cap V_2$ and we notice that $S_2$ is an MIS of $G_2$ (otherwise it contradicts the maximality of $S$). It follows that $|S_1|= \alpha_1$ and thus $S_1$ is an MIS of $G_1$ that contains $v$, which is a contradiction.\newline

Secondly, note that $\alpha(G_1+ G_2)=\max(\alpha(G_1)+\alpha(G_2))$ and that any independent set of $G_1+G_2$ is either a subset of $V_1$ or $V_2$ since all edges between $V_1$ and $V_2$ are present.

Assume that $G_1$ and $G_2$ are $1$-extendable and $\alpha(G_1)=\alpha(G_2)$, and let $v\in V_1$. By $1$-extendability of $G_1$, $v$ is in an independent set of $G_1$ of size $\alpha(G_1)$, and thus in an MIS of $G$. The same holds if $v\in V_2$. 

Suppose that $G_1+G_2$ is $1$-extendable. If $\alpha(G_1)<\alpha(G_2)$, then any vertex of $G_1$ is not in an MIS of $G$ since any independent set that contains that vertex is a subset of $V_1$. Thus, $\alpha(G_1)=\alpha(G_2)$, and let $v\in V_1$. By $1$-extendability of $G_1+G_2$, $v$ is in an independent set $S$ of size $\alpha(G)=\alpha(G_1)$. Since $S$ is either a subset of $V_1$ or $V_2$ and $v\in V_1\cap S$, we have $S\subseteq V_1$. It follows that $v$ is in an MIS of $G_1$ and that $G_1$ is $1$-extendable. The same holds for $G_2$.
\end{proof}

This characterization also leads to a linear-time algorithm for determining whether a cograph is 1-extendable.
We now focus on the complexity of determining the 1-extendable chromatic number of a cograph. We use this characterization in order to decide, in $k \cdot \alpha(G)^{O(k)}$ time, whether a cograph $G$ has 1-extendable chromatic number at most $k$. Later, we will prove that this number is actually always bounded by $\log_2(\alpha(G))+1$, hence proving that it can be determined in quasi-polynomial time. We leave open the question whether the problem is polynomial-time solvable, and argue that this might be a more difficult question that it looks like, even in the very restricted case of complete multipartite graphs.

\begin{theorem}\label{QPcograph}
Deciding whether a cograph $G$ has 1-extendable chromatic number at most $k$ can be done in time $k \cdot \alpha(G)^{O(k)}$.
\end{theorem}
\begin{proof}
A $k$-tuple of non-negative integers $(\alpha_1,...,\alpha_k)$ is \emph{feasible} for a graph $G$ if $G$ admits a 1-extendable $k$-partition $V_1,...,V_k$ such that $\alpha(G[V_i])=\alpha_i$ for any $1\leqslant i \leqslant k$. Given a graph $G$, we call $S(G)$ the set of feasible $k$-tuples of $G$. Note that $\cext(G)\leqslant k$ if, and only if, $S(G)\neq \emptyset$. We give here an algorithm that computes $S(G)$ for a cograph.

If $G$ is a single vertex, then $S(G) = {(e_i)}_{1\leqslant i \leqslant k}$ where $e_i$ is a $k$-tuple with a $1$ in the $i$-th position and $0$ elsewhere.

If $G=G_1+G_2$, then $(\alpha_1,...,\alpha_k)\in S(G)$ if, and only if, there exists $(\alpha_1^{(1)},...,\alpha_k^{(1)})\in S(G_1)$ and $(\alpha_1^{(2)},...,\alpha_k^{(2)})\in S(G_2)$ such that for any $1\leqslant i \leqslant k$, one of the three conditions is true 
\begin{itemize}
\item $\alpha_i^{(1)} = \alpha_i^{(2)} = \alpha_i$  ;
\item $\alpha_i^{(1)} = 0$ and $\alpha_i^{(2)}=\alpha_i$;
\item $\alpha_i^{(2)} = 0$ and $\alpha_i^{(1)}=\alpha_i$. 
\end{itemize}
The equivalence holds because if we consider a color class in a 1-extendable partition of a graph $G$ that intersects both $G_1$ and $G_2$, the independence numbers induced by both subgraphs must be the same, as stated in Proposition \ref{OpCographs}. Furthermore, if the color class does not intersect $G_1$, then the independence number of the color class restricted to $G_1$ is zero, and there are no restrictions on the independence number of the color class restricted to $G_2$. The same holds if the color class does not intersect $G_2$. We use the notation $\cap_0$ for this operation, \textit{i.e.} $S(G_1+G_2) = S(G_1)\cap_0S(G_2)$.

If $G=G_1\cup G_2$, then $S(G) = S(G_1)+S(G_2) = \{\overrightarrow{\alpha_1} + \overrightarrow{\alpha_2}, \overrightarrow{\alpha_1}\in S(G_1), \overrightarrow{\alpha_2}\in S(G_2)\}$, where addition is taken component by component. The result holds due to Proposition \ref{OpCographs} and the fact that if we consider the color class of a 1-extendable partition, denoted as $V'$, then the independence number of the whole graph $G$ restricted to $V'$ is equal to the sum of the independence numbers of $G_1$ and $G_2$ restricted to $V'$.

We now focus on the running time. At each step, we need to show that the set operations can be done efficiently. First observe that $S(G)$ can be stored using an array without any duplicates, and that $|S(G)|\leqslant (\alpha+1)^k$. 
With such an implementation, both $S(G_1)\cap_0 S(G_2)$ and $S(G_1)+S(G_2)$ can be computed in time $k|S(G_1)|\cdot |S(G_2)|$. In the first case, we can check for each pair $(\overrightarrow{\alpha_1},\overrightarrow{\alpha_2})\in S(G_1)\times S(G_2)$ if they satisfy the three properties of $S(G_1) \cap_0S(G_2)$ in time $k|S(G_1)|\cdot |S(G_2)|$. In the second case case, we can simply compute the pairwise sum of $S(G_1)$ and $S(G_2)$ in time $k|S(G_1)|\cdot |S(G_2)|$.

Let $C(\alpha)$ be the maximum number of operations of the algorithm on a cograph with independence number $\alpha$. We show by induction on the cograph $G$ that $C(\alpha) \leqslant k(\alpha+1)^{2k}$. 
If $G$ is a single vertex, the inequality holds, so we now focus on the disjoint union and complete sum. Assume $G=G_1\cup G_2$, and let $\alpha_1$ (resp. $\alpha_2$) be the independence number of $G_1$ (resp. $G_2$). Then the running time is at most:
\begin{eqnarray*}
C(\alpha) & \leqslant &  C(\alpha_1)+C(\alpha_2)+ k(\alpha_1+1)^k  (\alpha_2+1)^k \\
& \leqslant  & C(\alpha_1)+C(\alpha_2)+ k\left( \frac{\alpha}{2} + 1 \right)^{2k} \\
& \leqslant  & k(\alpha_1 +1)^{2k} + k(\alpha_2 +1)^{2k} + k\left( \frac{\alpha}{2} + 1 \right)^{2k} \\
& \leqslant  & 2k \left( \frac{\alpha}{2}+1 \right)^{2k} + k\left( \frac{\alpha}{2} + 1 \right)^{2k} \\
& \leqslant  & k (\alpha +1)^{2k}
\end{eqnarray*}

Where the second and fourth inequalities hold because $x_1^k + x_2^k \leqslant 2(\frac{x}{2})^k$ if $x = x_1+x_2$, and the last inequality holds whenever $k \geqslant 2$ and $\alpha \geqslant 2$ (if $k=1$ then one can decide in linear time, and $\alpha=1$ means that $G$ is a clique and thus the problem is obvious).
The case of a complete sum follows the same computations.
\end{proof}

\subsection{Generalization to bounded modular-width}

Here we show that the previous algorithm can be generalized to any graph class of bounded modular-width. In addition, we obtain an algorithm for testing whether a graph is 1-extendable which runs in single exponential time in the modular-width, and linear time in the number of vertices.

\subsubsection{Testing 1-extendablity}

Given a graph $G$, the \textsc{1-Extendability} problem asks whether $G$ is $1$-extendable. In any class of graphs where computing the independence number is polynomial, one can test $1$-extendability by verifying if $\alpha(G-N[v])= \alpha(G)-1$ for all $v\in V(G)$. Using this approach, testing $1$-extendability is FPT in classical parameters such as clique-width or tree-width. However, these algorithms are quadratic rather than linear. Here, we provide a linear FPT algorithm in modular-width, with a single exponential dependency on the parameter.
\medskip

The following three results extend Proposition \ref{OpCographs} to modular decomposition. Our approach is inspired by the algorithm for recognizing well-covered graphs which contain a few number of paths on four vertices~\cite{Klein2013}. The substitution case is of particular interest and is the main difference with the previous approach, where it is observed that if a graph is $1$-extendable, then its representative graph, which is a vertex-weighted graph where the weights correspond to the independence number of each strong module, is also $1$-extendable in a weighted sense. These results provide the foundation for developing a linear FPT algorithm based on the modular-width to effectively solve the \textsc{1-Extendability} problem.

Let $G$ be a graph together with its modular decomposition, and let $H$ be its representative graph (\textit{i.e.} the first node of the modular decomposition). We define the weighted representative graph $H_w$ \cite{Klein2013} by assigning to each node $t$ the weight $w(t) = \alpha(G[M_t])$, where $M_t$ is the corresponding module of $t$ in $G$. Observe that a weighted MIS of $H_w$ has weight $\alpha(G)$. A weighted graph is $1$-extendable if each vertex belongs to a maximum weighted independent set.

\begin{lemma}\label{LemmaModularDecomposition1}
If $G$ is $1$-extendable, then for each module $M$ of $G$, $G[M]$ is also $1$-extendable.
\end{lemma}

\begin{proof}
Suppose $G$ is $1$-extendable. Let $v$ be a vertex in $M$ and by $1$-extendability of $G$, there exists an MIS $I$ of $G$ that contains $v$. Let $I_M = I \cap M$ and we show by contradiction that $I_M$ is an MIS of $G[M]$.

If $I_M$ is not an MIS of $G[M]$, let $I_M'$ be such an MIS, and thus $|I_M'|>|I_M|$. By definition of a module, $I' = I \backslash I_M \cup I_M'$ is an independent set of $G$, and $|I'|>|I|$, contradicting the maximality of $I$. Therefore, $I_M$ is an MIS of $G[M]$ that contains $v$ and so $G[M]$ must be $1$-extendable.
\end{proof}

\begin{lemma}\label{LemmaModularDecomposition2}
If $G$ is $1$-extendable, then its weighted representative graph $H_w$ is $1$-extendable.
\end{lemma}

\begin{proof}
Let $t$ be a vertex of $H_w$ and let $M_t$ be the strong module of $G$ represented by $t$. By Lemma \ref{LemmaModularDecomposition1}, all strong modules of $G$ are $1$-extendable.

Let $v\in M_t$ and by $1$-extendability of $G$, $v$ is in an MIS $I$ of $G$. Let $t_1,...,t_k$ be the vertices of $H_w$ such that $I\cap M_{t_i} \neq \emptyset$ for all $1\leqslant i \leqslant k$. Note that $|I\cap M_{t_i}| =\alpha(G[M_{t_i}])$ for any $1\leqslant i\leqslant k$, since otherwise we could swap $I\cap M_{t_i}$ with a maximum independent set of $G[M_{t_i}]$, contradicting the assumption that $I$ is an MIS of $G$. 

It follows that $T=\{t_1,...,t_k\}$ is an independent set of $H$ of weight $\alpha(G)$, and thus a weighted MIS of $H_w$. Finally, $H_w$ is a weighted $1$-extendable graph.
\end{proof}

\begin{theorem}
A graph $G$ is 1-extendable if and only if all its maximal strong modules are $1$-extendable and if its weighted representative graph is $1$-extendable.
\end{theorem}

\begin{proof}
First, assume that $G$ is $1$-extendable, and let $M_t$ be any strong module of $G$. By Lemma \ref{LemmaModularDecomposition1}, $G[M_t]$ is also $1$-extendable. Then, by Lemma \ref{LemmaModularDecomposition2}, the weighted representative graph $H_w$ of $G$ is also $1$-extendable.

Conversely, suppose that all its maximal strong modules and $H_w$ are $1$-extendable. Let $v$ be any vertex in a strong module $M_t$ of $G$. Since $G[M_t]$ is $1$-extendable, there exists an MIS $I_t$ of $G[M_t]$ that contains $v$. Moreover, since $H_w$ is $1$-extendable, there exists an MIS ${t_1, \ldots, t_k}$ of $H_w$ that contains $t$. For each $i \in \{1, \ldots, k\}$ such that $t_i \neq t$, let $I_{t_i}$ be any MIS of $G[M_{t_i}]$. Then, the set $I = I_t \cup I_{t_1} \cup \cdots \cup I_{t_k}$ is an MIS of $G$ that contains $v$. Therefore, $G$ is $1$-extendable.
\end{proof}

\begin{theorem}\label{thm:test1ext-mw}
Deciding whether a graph $G$ is $1$-extendable can be done in time $2^{O(\mw(G))}\cdot n$.
\end{theorem}

\begin{proof}
The algorithm operates recursively on the modular decomposition of graph G to determine if it is 1-extendable and returns its independence number. The algorithm follows these steps:
\begin{itemize}
\item If $G = G_1 \cup G_2$, then $G$ is 1-extendable if and only if $G_1$ and $G_2$ are both 1-extendable by Proposition \ref{OpCographs}. In addition, $\alpha(G) = \alpha(G_1) + \alpha(G_2)$.
\item If $G = G_1 + G_2$, then $G$ is 1-extendable if and only if $G_1$ and $G_2$ are both 1-extendable by Proposition \ref{OpCographs}, and $\alpha(G_1) = \alpha(G_2)$. In addition, $\alpha(G) = \max(\alpha(G_1), \alpha(G_2))$.
\item If $G = H[G_1, \dots, G_m]$ with $m \leqslant \mw(G)$, then we apply the brute-force algorithm on $H$ to check if $G$ is 1-extendable and return $\alpha(G)$ in time $2^m \leqslant 2^{\mw(G)}$.
\end{itemize}

Let $G$ be an $n$-vertex graph. The tree resulting from the modular decomposition of $G$ has $n$ leaves, and each internal node has a degree of at least 3 (except the root). Thus, the tree has at most $n$ internal nodes. Since the algorithm performs at most $2^{\mw(G)}$ operations on the nodes of the tree, the overall running time of the algorithm is at most $2^{\mw(G)} \cdot n$.
\end{proof}

\subsubsection{1-extendable partition}
We can use the modular decomposition to construct an algorithm that solves the \textsc{1-Extendable Partition} problem in polynomial time when parameterized by the number of parts and the modular-width of the input graph.

\begin{theorem}\label{ExpModDecomp}
Deciding whether a graph $G$ with $n$ vertices and independence number $\alpha$ has 1-extendable chromatic number at most $k$ can be done in time $\alpha^{O(\mw(G)k)}\cdot n$.
\end{theorem}

\begin{proof}
The algorithm is an extension of the algorithm of Theorem \ref{QPcograph} with the substitution case of the modular decomposition. We recall that $S(G)$ is the set of feasible $k$-tuples of $G$ (see Theorem \ref{QPcograph}). We give here an algorithm that computes $S(G)$ for any graph $G$.

\begin{itemize}
\item If $G$ is a single vertex $S(G) = \{(e_i)\}_{1\leqslant i \leqslant k}$.
\item If $G=G_1\cup G_2$, then $S(G)=S(G_1)+S(G_2)$.
\item If $G=G_2+G_2$, then $S(G)=S(G_1)\cap_0 S(G_2)$.
\item If $G=H[G_1,...,G_m]$ where $H$ is the representative graph. We try to compute $S(G)$ from $S(G_1),...,S(G_m)$. We have $(\alpha_1,...,\alpha_k)\in S(G)$ if, and only if, for any $1\leqslant j \leqslant m$, there exists $(\alpha_1^{(j)},\dots,\alpha_k^{(j)}) \in S(G_j)$ such that for any $1\leqslant i \leqslant k$, the weighted representative graph $H_{w_i}$ with weights $(\alpha_i^{(1)},\dots,\alpha_i^{(m)})$ is 1-extendable and $\alpha(H_{w_i}) = \alpha_i$.
\end{itemize}

We now focus on the running time. We show that, at any iteration of the algorithm, the number of operations is at most $2^{\mw(G)}(\alpha+1)^{k\mw(G)}$, where $\alpha$ is the independence number of the current graph $G$.  
\begin{itemize}
\item If $G$ is a single vertex, then the result holds.
\item If $G=G_1\cup G_2$, the running time is 
$$
2k|S_1|\cdot |S_2| \leqslant \left( \frac{\alpha}{2}+1\right)^{2k} \leqslant 2^{\mw(G)}\cdot (\alpha+1)^{\mw(G)k} 
$$
since $\mw(G) \geqslant 2$.
\item If $G=G_1+G_2$, the case is similar to the previous one. 
\item If $G=H[G_1,\dots,G_m]$, we assume $S(G_1),\dots,S(G_m)$ already computed. To compute $S(G)$, we need to enumerate $S(G_1)\times\dots\times S(G_m)$ and for each color class check if the weighted representative graph is 1-extendable using the brute-force algorithm in time $2^m$. Thus, the running time $C$ is at most:
$$
C \leqslant \left(\prod_{i=1}^m |S(G_i)|\right)\cdot  k\cdot 2^{m}
$$
However, we can bound $m$ by $\mw(G)$ and let $\alpha_i = \alpha(G_i)$. Since $G_i$ is an induced subgraph of $G$, we have $\alpha_i \leqslant \alpha$. It follows that:
\begin{eqnarray*}
C & \leqslant &  \left(\prod_{i=1}^m (\alpha_i+1)^k\right)\cdot  k\cdot 2^{m} \\
& \leqslant  & \left(\prod_{i=1}^m (\alpha+1)^k\right)\cdot   k\cdot 2^{m} \\
& =  & \left(\alpha +1\right)^{mk} \cdot k \cdot 2^{\mw(G)} \\
& \leqslant  & (\alpha+1)^{\mw(G)k} \cdot k \cdot 2^{\mw(G)} \\
\end{eqnarray*}
\end{itemize}

Finally, let $G$ be the input graph. Note that the decomposition tree has at most $n-1$ internal nodes, and that the current graph has modular-width at most $\mw(G)$ at each iteration of the algorithm. Thus the overall running time is at most 
$$
(\alpha+1)^{\mw(G)k} 2^{\mw(G)}k\cdot n
$$
\end{proof}

\subsection{Complete multipartite graphs}

We continue our exploration of the partition into $1$-extendable induced subgraphs by focusing on a specific class of cographs, namely the complete multipartite graphs. Although being a very restricted case, it appears that the problem within this graph class establishes connections with other fields, making it inherently interesting. We first begin with the following observation.

\begin{proposition}
A complete multipartite graph is 1-extendable if, and only if, it is balanced.
\end{proposition}

\begin{proof}
Let $G$ be a complete multipartite graph with parts $V_1,...,V_m$. $G$ is a cograph obtained with $m$ consecutive joins between independent sets. By Proposition \ref{OpCographs}, $G$ is $1$-extendable if and only if all parts have same cardinality.
\end{proof}

With that previous proposition, partitioning a complete multipartite graph into 1-extendable induced subgraphs is exactly the same as partitioning into balanced complete multipartite subgraphs. 
An example is illustrated in~Figure~\ref{FigMultiparti}.

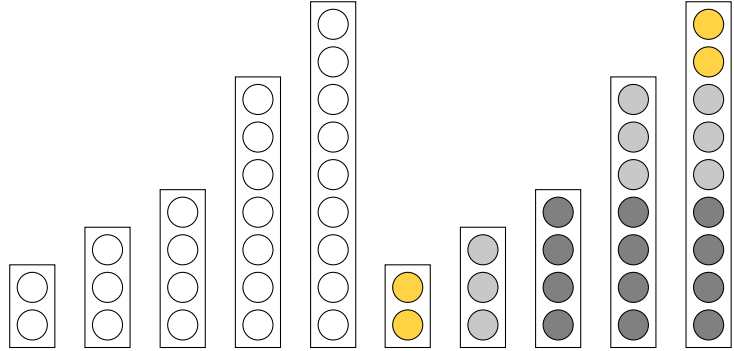
\begin{figure}[t]
\centering
\begin{minipage}{0.4\linewidth}
\centering
\begin{tikzpicture}

\foreach \j in {0,1}{
	\draw (0,\j/2) circle (0.2);}
\foreach \j in {0,...,2}{
	\draw (1,\j/2) circle (0.2);}
\foreach \j in {0,...,3}{
	\draw (2,\j/2) circle (0.2);}
\foreach \j in {0,...,6}{
	\draw (3,\j/2) circle (0.2);}
\foreach \j in {0,...,8}{
	\draw (4,\j/2) circle (0.2);}

\draw (-0.3,-0.3) rectangle (0.3,0.8);
\draw (1-0.3,-0.3) rectangle (1+0.3,1.3);
\draw (2-0.3,-0.3) rectangle (2+0.3,1.8);
\draw (3-0.3,-0.3) rectangle (3+0.3,3.3);
\draw (4-0.3,-0.3) rectangle (4+0.3,4.3);
\end{tikzpicture}
\end{minipage}
\begin{minipage}{0.4\linewidth}
\centering
\begin{tikzpicture}
\foreach \j in {0,1}{
	\draw[fill=lipicsYellow] (0,\j/2) circle (0.2);
	\draw[fill=lipicsYellow] (4,\j/2+7/2) circle (0.2);}

\foreach \j in {0,...,2}{
	\draw[fill=lipicsLightGray] (1,\j/2) circle (0.2);
	\draw[fill=lipicsLightGray] (3,{\j/2+2}) circle (0.2);
	\draw[fill=lipicsLightGray] (4,{\j/2+2}) circle (0.2);}

\foreach \j in {0,...,3}{
	\draw[fill=lipicsGray] (2,\j/2) circle (0.2);
	\draw[fill=lipicsGray] (3,\j/2) circle (0.2);
	\draw[fill=lipicsGray] (4,\j/2) circle (0.2);}

\draw (-0.3,-0.3) rectangle (0.3,0.8);
\draw (1-0.3,-0.3) rectangle (1+0.3,1.3);
\draw (2-0.3,-0.3) rectangle (2+0.3,1.8);
\draw (3-0.3,-0.3) rectangle (3+0.3,3.3);
\draw (4-0.3,-0.3) rectangle (4+0.3,4.3);
\end{tikzpicture}
\end{minipage}
\caption{A complete multipartite graph (left) and a $1$-extendable $3$-partition of it, represented as a coloring (right). For the sake of readability, the edges are not shown, the rectangles represent the partition into independent sets.} \label{FigMultiparti}
\end{figure}

\medskip

To avoid confusion, from now on the term \textit{part} is reserved for the elements of the (initial) partition of the complete multipartite graph, while we will use the term \textit{color} to refer to the elements of a $1$-extendable partition.

In a complete multipartite graph $G$ with partition $V_1$, $\dots$, $V_q$, every vertex of a same part has the same neighborhood, hence a 1-extendable partition $(C_1, \dots, C_k)$ in that case can actually be defined simply by giving, for each color $C_i$, its width together with the subset of parts of $G$ it intersects. Alternatively, one can also give the way each part $V_i$ is split by the 1-extendable partition. And since all vertices of a same part have the same neighborhood, it is actually sufficient to give the way $|V_i|$ is decomposed by the 1-extendable partition.
This observation leads to an equivalent\footnote{Naturally, in our graph-theoretical version, the integers are encoded in unary.} reformulation of the problem as a numbers problem. In the usual sense, a set $B$ of integers is a \textit{generating set} of a set $N$ of integers if every $p \in N$ can be written as a sum of elements in $B$. This problem is thus called \newsubsetsum.
\medskip

\noindent\fbox{
\begin{minipage}{.97\linewidth}
\newsubsetsum\\
\textbf{Input:} $m$ integers $n_1,...,n_m$ and an integer $k$.\\
\textbf{Question :}  Are there $k$ integers $\alpha_1,...,\alpha_k$ such that, for any $1\leqslant i \leqslant m$, there exists $J\subset [k]$ such that $\sum_{j\in J} a_j = n_i$ ?
\end{minipage}
}\newline

As observed previously, this problem is equivalent to the \textsc{1-Extendable Partition} problem in complete multipartite graphs. 
It appears that we are able to obtain a much simpler quasi-polynomial algorithm than in the case of cographs. Consider an instance $n_1,...,n_m,k$ of the \newsubsetsum problem where the input is encoded in unary. We call $\alpha = \max_{1\leqslant i \leqslant m} n_i$ and $n  =\sum_{i=1}^m n_i$. The algorithm lists all the $k$-tuples $(\alpha_1,...,\alpha_k) \in \{0,...,\alpha\}^k$, and solves, for each $i \in [m]$, the \textsc{Subset Sum} problem with numbers $(\alpha_1,...,\alpha_k)$ and target $n_i$ in time $O(k^2\alpha)$. The overall complexity of such an algorithm is thus $O(\alpha^{k+1} k^2)$. 
Now observe that we may always assume that $k \leqslant \lceil \log_2(\alpha) \rceil +1 $, since $\{2^i : 0 \leqslant i \leqslant \lceil \log_2(\alpha) \rceil\}$ is always a solution, by taking the binary representation of each integer $n_i$.

It turns out that a number-theoretic version of \newsubsetsum was already studied in the context of cryptanalysis, called \textsc{Hidden Subset Sum}~\cite{Boyko,Nguyen,Coron}. More precisely, a fast generator of random pairs $(x, g^x \mod p)$ was first presented~\cite{Boyko}, whose security relies on the potential hardness of  \textsc{Hidden Subset Sum}. However, an efficient lattice-based attack was quickly presented~\cite{Nguyen}, and further improved~\cite{Coron}. In \cite{collins2007nonnegative}, the problem was shown to be NP-complete. However, the existence of a (worst-case) polynomial-time algorithm for \newsubsetsum is still open, the input being encoded in unary.

\section{Extremal results}\label{sec:extremal}

In this section, we explore structural and extremal properties of 1-extendable partitions. Specifically, we show that $\cext(G) = O(\sqrt{n})$ for any $n$-vertex graph $G$. Then, we focus once again on cographs and demonstrate that $\cext$ is always logarithmic in the size of the independence number, and we provide two constructions that achieve this optimal bound. The first one is a complete multipartite graph and the second one is both an interval graph and a cograph. Together with the algorithm described in Section~\ref{sec:complexity}, it implies that \textsc{1-Extendable Partition} can be solved in quasi-polynomial time in cographs.

\subsection{General case}

\begin{lemma}
For any graph $G$, $\cext(G) \leqslant \alpha(G)$.
\end{lemma}

\begin{proof}
We show that $\cext(G)\leqslant \alpha(G)$ by induction on $\alpha(G)$. If $\alpha(G)=1$, then $G$ is a clique and thus $G$ is $1$-extendable. We assume that the result is true for any graph $G$ such that $\alpha(G) \leqslant k$ for some $k\geqslant 1$. Let $G$ be a graph such that $\alpha(G) =k+1$, and let $V_{k+1}$ be the set of all vertices of $G$ that are in an MIS. We notice that $G[V_{k+1}]$ is $1$-extendable since each vertex is in an MIS by definition. In addition, $G[V(G) - S_{k+1}]$ has an independence number at most $k$ since all independent sets of size $k+1$ have been removed. By applying the induction hypothesis, $V(G) - V_{k+1}$ can be partioned into $k$ subsets $V_1,...,V_k$ (possibly empty) such that $G[V_i]$ is $1$-extendable for any $1\leqslant i \leqslant k$. By completing the partition with $V_{k+1}$, we obtain a desired partition for $V(G)$.
\end{proof}

The previous bound can be combined to obtain a square root upper bound in the order of a graph.

\begin{theorem}\label{thm:sqrtbound}
For any graph $G$ with $n$ vertices, 
$$
\cext(G) \leqslant 2\sqrt{n}
$$
\end{theorem}

\begin{proof}
Let $G$ be a graph. We consider the following greedy coloring of $G$. Select an MIS $S_1$ of $G$ and repeat with $G-S_1$. We obtain a partition $S_1,...,S_k$ of $G$ where each $G[S_i]$ is an independent set. In addition, we notice that for any $1\leqslant i<k$, $|S_i| \geqslant |S_{i+1}|$, and let $k_0$ be the largest integer $i$ such that $|S_i|\geqslant \sqrt{n}$. We notice that $k_0\leqslant \sqrt{n}$, otherwise we would obtain strictly more than $n$ vertices in $S_1\cup ...\cup S_{k_0}$. Let $V_1 = \bigcup_{i : |S_i|\geqslant \sqrt{n}} S_i$, and $V_2=V\backslash V_1$. First, we have $\chi(G[V_1]) \leqslant k_0 \leqslant \sqrt{n}$. Then, we have $\alpha(G[V_2]) < \sqrt{n}$, since if there is a stable $S\subseteq V_2$ such that $|S|\geqslant \sqrt{n}$, it would have been selected instead of $S_{k_0+1}$ in the greedy algorithm since $|S_{k_0+1}|<\sqrt{n}$.

By coloring separatly $G[V_1]$ and $G[V_2]$, we obtain a $1$-extendable coloring of $G$ using at most $2\sqrt{n}$ colors. 
\end{proof}

\subsection{Cographs}

\subsubsection{Upper bound}

In this section, we prove that the 1-extendable chromatic number of a cograph is at most logarithmic in the size of its independence number. The general idea of the proof is to extract a 1-extendable induced subgraph while reducing the independence number of the remaining graph by a factor of 2. 

We first need a technical lemma regarding the partition of integers. It will be used in the disjoint union case, where we need to ensure that the extracted 1-extendable induced subgraph have the same independence number in both parts, in order to choose the right value for this independence number.

\begin{lemma}\label{DecompositionLemma}
For any $\alpha_1,\alpha_2\geqslant 0$ and $k \in \{0,...,\alpha_1+\alpha_2\}$, there exist $k_1,k_2\geqslant 0$ such that :
\begin{itemize}
\item $k_1\leqslant \alpha_1$ and $k_2\leqslant \alpha_2$ ;
\item $k_1+k_2=k$ ;
\item $\max(k_1-1,\alpha_1-k_1)+\max(k_2-1,\alpha_2-k_2)\leqslant \max(k-1,\alpha_1+\alpha_2-k)$.
\end{itemize}
\end{lemma}

\begin{proof} Let $\alpha = \alpha_1+\alpha_2$. We consider a first case when:
$$
k\geqslant \left\lceil  \frac{\alpha_1+1}{2}\right\rceil + \left\lceil  \frac{\alpha_2+1}{2} \right\rceil =: k_0^+
$$
We consider the set 
$$
S^+ = \left\lbrace k_1+k_2 \mid \left\lceil  \frac{\alpha_1+1}{2}\right\rceil  \leqslant k_1 \leqslant \alpha_1, \left\lceil  \frac{\alpha_2+1}{2}\right\rceil  \leqslant k_2 \leqslant \alpha_2\right\rbrace 
$$ 
and we note that $S^+$ is exactly the interval $
\{ k_0^+, \dots, \alpha_1+\alpha_2 \}
$ and thus contains $k$. By taking the couple $k_1,k_2$ which sums to $k$ in $S^+$, we have $k_1+k_2=k$ and $\max(k_{\ell}-1,\alpha_{\ell}-k_{\ell}) = k_{\ell} -1$ for $\ell \in \{1,2\}$, and $k_1-1+k_2-1 = k-2\leqslant \max(k-1,\alpha-k)$.
We can make a similar proof when: 
$$
k\leqslant \left\lfloor  \frac{\alpha_1+1}{2}\right\rfloor + \left\lfloor  \frac{\alpha_2+1}{2} \right\rfloor =: k_0^-
$$
by taking 
$$
S^- = \left\lbrace k_1+k_2 \mid 0 \leqslant k_1 \leqslant \left\lfloor  \frac{\alpha_1+1}{2}\right\rfloor, 0 \leqslant k_2 \leqslant \left\lfloor  \frac{\alpha_2+1}{2}\right\rfloor\right\rbrace
$$ 
The only critical case happens when $k_0^- < k < k_0^+$, which happens if, and only if, $\alpha_1$ and $\alpha_2$ are both even. In this case, we have $k =  k_0^- + 1 = k_0^- -1$, and thus $k= \alpha_1/2+\alpha_2/2 + 1 = \alpha/2 +1$. By taking $k_1 = \alpha_1/2+1$ and $k_2=\alpha_2/2$, we have have $\max(k_1-1, \alpha_1-k_1)= \alpha_1/2$ and $\max(k_2-1, \alpha_2-k_2) = \alpha_2/2$. Then, $\alpha_1/2+\alpha_2/2 = \alpha/2 \leqslant \max(k-1, \alpha-k)$ which concludes the proof.
\end{proof}

The next lemma is the key ingredient of our upper bound.

\begin{lemma}\label{TheoremCograph}
For any cograph $G$ and any $k\in \{0,...,\alpha(G)\}$, there exists a partition of the vertices into two subsets $V_1$ and $V_2$ such that 
\begin{itemize}
\item $G[V_1]$ is 1-extendable ;
\item $\alpha(G[V_1])=k$ ;
\item  $\alpha(G[V_2]) \leqslant \max(k-1,\alpha(G)-k)$.
\end{itemize}
\end{lemma}
\begin{proof}
We prove the result by structural induction on cographs. If the cograph is a single vertex $v$, then the property trivially holds. 

If $G=G_1 \cup G_2$. Let $k\in \{0,...,\alpha(G)\}$ and notice that $\alpha(G) = \alpha(G_1)+\alpha(G_2)$.
By induction hypothesis on $G_i$ ($i\in \{1,2\}$), we can find a partition $(V_1^{(i)}, V_2^{(i)})$ of $V(G_i)$ such that
\begin{itemize}
\item $G_i[V_1^{(i)}]$ is 1-extendable ;
\item $\alpha(G_i[V_1^{(i)}])=k_i$ ;
\item $\alpha(G_i[V_2^{(i)}]) \leqslant \max(k_i-1,\alpha_i-k_i)$.
\end{itemize}
where $k_1$ and $k_2$ are obtained using Lemma \ref{DecompositionLemma}. Let $V_1= V_1^{(1)} \cup V_1^{(2)}$ and $V_2= V_2^{(1)} \cup V_2^{(2)}$. Since $G_1[V_1^{(1)}]$ and $G_2[V_1^{(2)}]$ are 1-extendable, $G[V_1]$ is also 1-extendable. In addition, 
$$
\alpha(G[V_1]) = \alpha(G_1[V_1^{(1)}]) + \alpha(G_2[V_1^{(2)}])=k_1+k_2=k
$$
Then, 
\begin{align*}
\alpha(G[V_2]) &= \alpha(G_1[V_2^{(1)}])+  \alpha(G_2[V_2^{(2)}])  \\
&\leqslant \max(k_1-1,\alpha_1-k_1) + \max(k_2-1,\alpha_2-k_2) \\
&\leqslant \max(k-1,\alpha-k)
\end{align*}

The remaining case is $G=G_1+G_2$. Let $k\in \{0,...,\alpha(G)\}$, where $\alpha(G) = \max(\alpha_1,\alpha_2)$. Without loss of generality, we assume $\alpha_1\geqslant \alpha_2$. We distinguish two cases. 

\noindent First, we assume that $k\leqslant \alpha_2$, and we apply the induction hypothesis on $G_1$ and $G_2$ with $k$, and get a partition $(V_1^{(i)}, V_2^{(i)})$ of $V(G_i)$, $i \in \{1, 2\}$ as previously. The subgraph $G[V_1^{(1)} + V_1^{(2)}]$ is 1-extendable of independence number $k$.  Then,
\begin{align*}
\alpha(G[V_2^{(1)} + V_2^{(2)}]) &= \max(\alpha(G_1[V_2^{(1)}]),\alpha(G_2[V_2^{(2)}]) \\
&\leqslant \max(\max(k-1,\alpha_1-k), \max(k-1,\alpha_2-k))\\
&= \max(k-1, \alpha_1-k, \alpha_2-k) \\
&= \max(k-1, \max(\alpha_1,\alpha_2) -k ) \\
&= \max(k-1, \alpha -k)
\end{align*}
Secondly, we assume that $k>\alpha_2$. We apply the induction hypothesis on $G_1$ with $k$ and obtain a partition $(V_1^{(1)}, V_2^{(1)})$. Thus, $G[V_1^{(1)}]$ is 1-extendable with independence number $k$. Then, 
\begin{align*}
\alpha(G[V_1^{(2)}+V_2]) &= \max(\alpha(G_1[V_2^{(1)}]), \alpha(G_2))) \\
&\leqslant \max(\max(\alpha_1-k, k-1), k-1) \\
&= \max(\alpha-k,k-1)
\end{align*}
That ends the induction. 
\end{proof}
%

We derive the upper bound by recursively applying the previous lemma.

\begin{theorem}\label{LogarithmicBoundCographs}
For any cograph $G$, $\cext(G) \leqslant\log_2(\alpha(G))+1$.
\end{theorem}

\begin{proof}
We show the result by induction on $\alpha(G)$.
If $\alpha(G)=1$, then the graph is a clique and thus may be colored using $1$ color.
If not, let $k= \lceil ~\alpha(G)/2 ~\rceil$ and using Lemma \ref{TheoremCograph}, we can partition the vertices of $V$ into two subsets $V_1$ and $V_2$ such that $G[V_1]$ is 1-extendable and
$$
\alpha(G[V_2]) \leqslant\left\lfloor \frac{\alpha(G)}{2} \right\rfloor
$$
We only need one color for $G[V_1]$, and we recursively color $G[V_2]$ using $\log_2(\lfloor \alpha /2\rfloor)+1 \leqslant \log_2(\alpha)$ colors. Altogether, we can color $G$ using $\log_2(\alpha)+1$ colors such that each color class induces a 1-extendable graph.
\end{proof}

The bound in the previous theorem is obtained through recursive calls that halve the independence number, resulting in a logarithmic number of colors. However, this may not always be the optimal choice. For example, if $G$ is an independent set, a better choice would be to use only one color for the entire graph. Nevertheless, as we will see in the next section, the bound obtained from Theorem \ref{LogarithmicBoundCographs} is structurally optimal, as there exist cographs for which at least a logarithmic number of colors is required.

Combining the previous result with the algorithm presented in Theorem~\ref{QPcograph} allows to obtain the following.

\begin{corollary}
There exists an algorithm solving \textsc{1-Extendable Partition} on cographs in quasi-polynomial time.
\end{corollary}

\begin{proof}
Using Theorem \ref{LogarithmicBoundCographs}, $\cext(G) \leqslant \log_2(n)+1$. Thus, the following algorithm solves the problem :
\begin{itemize}
\item if $k\geqslant \log_2(n)+1$, the instance is a yes-instance ;
\item if $k<\log_2(n)+1$ runs the previous algorithm in time $\mathcal{O}(n^{2k+1})$.
\end{itemize}

This algorithm solves \textsc{1-Extendable Partition} on cographs in time at most $\mathcal{O}(n^{2\log_2(n)+3})$, which is quasi-polynomial in $n$. 
\end{proof}

Observe that Theorem~\ref{LogarithmicBoundCographs} can be generalized to every graph whose vertex set can be partitioned into a bounded number of induced cographs. This is for instance the case of $P_4$-bipartite graphs which include distance-hereditary graphs, parity graphs, $P_4$-reducible graphs, and $P_4$-sparse graphs.

\subsubsection{Lower bound}

We show that the previous logarithmic bound is essentially tight by constructing two different families of cographs that achieve this bound.

\paragraph{Complete multipartite graph.}

We consider the graph $G_k$, where $k\geqslant 0$, which is the complete multipartite graph with parts $V_0,...,V_k$ and $|V_i|=2^i$ for any $0\leqslant i \leqslant k$.

\begin{theorem}\label{ExtremalCompleteMultipartite}
For any $k\geqslant 0$, $G_k$ has $2^{k+1}-1$ vertices and $\cext(G_k)\geqslant k+1$.
\end{theorem}

\begin{proof}
Let $G_k$ be a complete multipartite graph with parts sets $V_0,...,V_k$ and $|V_i|=2^i$ for any $0\leqslant i \leqslant k$. For any $k\geqslant 0$, we easily show that $G_k$ has $2^{k+1}-1$ vertices. By contradiction, assume that there exists some 1-extendable $k$-coloring of $G$ and let $\alpha_{0},\dots,\alpha_{k-1}$ be the width of each of the $k$ balanced complete multipartite graphs induced by each color class. In particular, we assume that the widths are sorted by apparition order of the corresponding colors in $V_0,\cdots V_k$. We show by induction on $i$ that $\alpha_i \leqslant 2^i$.
\begin{itemize}
\item $V_0$ is composed of a unique vertex $v$, and so the color class of $v$ has width at most $|V_0|=1$. Thus, we have $\alpha_0=1$.
\item We suppose $\alpha_j \leqslant 2^j$ for any $j\leqslant i-1$ for some $i\geqslant 1$. The set $V_i$ cannot be colored using only the $i-1$ first colors. Indeed, we have $\alpha_0+\dots +\alpha_{i-1} \leqslant 2^{i}-1$ by induction hypothesis, and $|V_i| = 2^i$. Thus, a new color is used for the set $V_i$, which implies that $\alpha_i\leqslant |V_i|=2^i$. 
\end{itemize}
Thus, we have $\alpha_0+...+\alpha_{k-1} \leqslant 2^{k}-1$, and so the last set $V_k$ cannot be entirely colored using only the $k$ colors, which is a contradiction. Finally, $\cext(G_k) \geqslant k+1$.
\end{proof}

\paragraph{Interval graph.}

A pertinent question to consider is whether there exist cographs that require a logarithmic number of colors without being complete multipartite. In this context, we present an example that shows the existence of such cographs. The construction will be based on the creation of a specific interval graph denoted as $G_k$ ($k\geqslant 1$), defined as follows:
\begin{itemize}
\item $G_1 = K_1$, i.e., the graph containing only a single vertex;
\item For any $i\geqslant 1$, $G_{k+1} = K_1 + (G_k \cup G_k)$, where $K_1$ denotes a single vertex and $(G_k \cup G_k)$ represents the disjoint union of two copies of $G_k$.
\end{itemize}
It can be observed that for any $k\geqslant 1$, in addition of being a cograph, the graph $G_k$ is an interval graph, as depicted in Figure \ref{IntervalGraph}. By immediate induction, it may be shown that $\alpha(G_k)=2^k-1$ for any $k$.

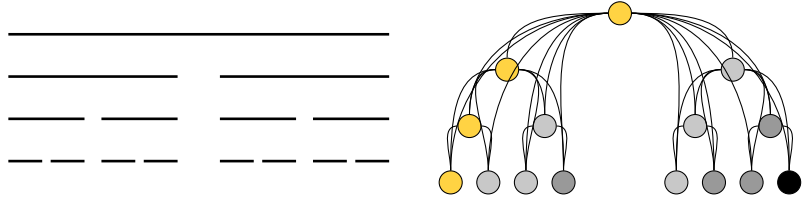
\begin{figure}[t]

\centering
\begin{minipage}{0.45\linewidth}
\centering
\begin{tikzpicture}[scale=0.45]
\draw [, line width=1pt](6.25,8.5) to (17.5,8.5);
\draw [, line width=1pt](6.25,7.25) to (11.25,7.25);
\draw [, line width=1pt](12.5,7.25) to (17.5,7.25);
\draw [, line width=1pt](6.25,6) to (8.5,6);
\draw [, line width=1pt](9,6) to (11.25,6);
\draw [, line width=1pt](12.5,6) to (14.75,6);
\draw [, line width=1pt](17.5,6) to (15.25,6);
\draw [, line width=1pt](8.5,4.75) to (7.5,4.75);
\draw [, line width=1pt](9,4.75) to (10,4.75);
\draw [, line width=1pt](11.25,4.75) to (10.25,4.75);
\draw [, line width=1pt](6.25,4.75) to (7.25,4.75);
\draw [, line width=1pt](12.5,4.75) to (13.5,4.75);
\draw [, line width=1pt](14.75,4.75) to (13.75,4.75);
\draw [, line width=1pt](15.25,4.75) to (16.25,4.75);
\draw [, line width=1pt](17.5,4.75) to (16.5,4.75);
\end{tikzpicture}
\end{minipage}
\begin{minipage}{0.45\linewidth}
\centering
\begin{tikzpicture}
\node[draw,circle, fill = lipicsYellow] (0) at (0,2.25)  {} ;
\node[draw,circle, fill = lipicsYellow] (1) at (-1.5,1.5) {} ;
\node[draw,circle, fill = lipicsLightGray] (2) at (1.5,1.5)  {} ;
\node[draw,circle, fill = lipicsYellow] (3) at (-2,0.75) {} ;
\node[draw,circle, fill = lipicsLightGray] (4) at (-1,0.75) {} ;
\node[draw,circle, fill = lipicsLightGray] (5) at (1,0.75) {} ;
\node[draw,circle, fill = lipicsGray!80] (6) at (2,0.75)  {} ;
\node[draw,circle, fill = lipicsYellow] (7) at (-2.25,0) {} ;
\node[draw,circle, fill = lipicsLightGray] (8) at (-1.75,0) {} ;
\node[draw,circle, fill = lipicsLightGray] (9) at (-1.25,0) {} ;
\node[draw,circle, fill = lipicsGray!80] (10) at (-0.75,0)  {} ;
\node[draw,circle, fill = lipicsLightGray] (11) at (0.75,0) {} ;
\node[draw,circle, fill = lipicsGray!80] (12) at (1.25,0) {} ;
\node[draw,circle, fill = lipicsGray!80] (13) at (1.75,0) {} ;
\node[draw,circle, fill = black] (14) at (2.25,0)  {} ;

\foreach \i in {1,3,4,7,8,9,10}{
	\draw[-, opacity = 0.5] (0) to[out=180,in=90] (\i) ;
};
\foreach \i in {2,5,6,11,12,13,14}{
	\draw[-, opacity = 0.5] (0) to[out=0,in=90] (\i) ;
};
\foreach \i in {3,7,8}{
	\draw[-, opacity = 0.5] (1) to[out=180,in=90] (\i) ;
};
\foreach \i in {4,9,10}{
	\draw[-, opacity = 0.5] (1) to[out=0,in=90] (\i) ;
};
\foreach \i in {5,11,12}{
	\draw[-, opacity = 0.5] (2) to[out=180,in=90] (\i) ;
};
\foreach \i in {6,13,14}{
	\draw[-, opacity = 0.5] (2) to[out=0,in=90] (\i) ;
};

\draw[-, opacity = 0.5] (3) to[out=180,in=90] (7) ;
\draw[-, opacity = 0.5] (3) to[out=0,in=90] (8) ;
\draw[-, opacity = 0.5] (4) to[out=180,in=90] (9) ;
\draw[-, opacity = 0.5] (4) to[out=0,in=90] (10) ;
\draw[-, opacity = 0.5] (5) to[out=180,in=90] (11) ;
\draw[-, opacity = 0.5] (5) to[out=0,in=90] (12) ;
\draw[-, opacity = 0.5] (6) to[out=180,in=90] (13) ;
\draw[-, opacity = 0.5] (6) to[out=0,in=90] (14) ;
\end{tikzpicture}
\end{minipage}
\caption{An interval representation of $G_4$, and a 1-extendable $4$-partition of it. In the partition, each color class induces a disjoint union of cliques, which is $1$-extendable.}\label{IntervalGraph}
\end{figure}

We start by a quite simple but yet useful lemma, saying that if a 1-extendable graph has a vertex linked to all others, the graph must be a clique.

\begin{lemma}
Let $G=(V,E)$ be a 1-extendable graph. If there exists $v\in V$ such that for any $u\in V\backslash \{v\}$, $vu\in E$, then $G$ is a clique.
\end{lemma}

\begin{proof}
Assume that $G$ is not a clique and thus $\alpha(G) \geqslant 2$. Thus $v$ cannot be in an independent set of size at least $2$, otherwise it would be connected to the other vertices of the independent set.
\end{proof}

\begin{theorem}
For any $k\geqslant 0$, $\cext(G_k)\geqslant k$.
\end{theorem}

\begin{proof}
We prove the result by induction on $k$. For $k=1$, the result is trivial. We assume the result to be true for some $k$. By contradiction, suppose $\cext(G_{k+1}) \leqslant k$. We write $G_{k+1} = {v} + (H_1 \cup H_2)$, where $v$ is the vertex at the top of $G_{k+1}$, and $H_1$ and $H_2$ are two copies of $G_k$. Let $c : V(G_{k+1}) \rightarrow [k]$ be a 1-extendable $k$-coloring of $G_{k+1}$, and without loss of generality, assume that $c(v)=k$. Since $v$ is connected to every other vertex of $G_{k+1}$, its color class forms a clique. Hence, there exists some $i \in \{1,2\}$ such that $c(V(H_{i})) \subseteq [k-1]$, meaning that $c$ does not use the color $k$ on one of the copies of $G_k$. If it used the color $k$ on both copies, the color class of $v$ would not form a clique. Without loss of generality, we assume that the copy $H_1$ is the one not using the color $k$. Thus, $c(V(H_1))$ is a 1-extendable $(k-1)$-coloring of $G_k$, which leads to a contradiction. Hence, we conclude that $\cext(G_{k+1})\geqslant k+1$.
\end{proof}

Interestingly, this logarithmic lower bound is also tight for interval graphs, since every interval graph $G$ satisfies $\chi_s(G) \leqslant \log_2(n) +1$, \textit{i.e.} we can partition its vertex set into a logarithmic number of parts, each of them inducing a disjoint union of cliques~\cite{Broersma2002}.

\section{Conclusion and open questions} \label{sec:conclusion}

Motivated by practical applications in wireless networks, this study introduced the concept of the 1-extendable chromatic number and examined its structural and algorithmic properties.

Our primary contribution is a quasi-polynomial algorithm for determining this number in cographs, achieved by combining an algorithm exponential in the number of available colors with a logarithmic upper bound. The existence of a polynomial-time algorithm remains an open question, even in the very restricted case of complete multipartite graphs, where the problem intriguingly reformulates in number-theoretic terms. Additionally, we extended our algorithm to arbitrary graphs using modular decomposition, resulting in a polynomial-time algorithm when both the modular width and the number of parts are fixed. We also provided a linear FPT algorithm in terms of modular width to determine if an input graph is $1$-extendable, with a single exponential dependency on the parameter.

On the structural side, we demonstrated the existence of $n$-vertex graphs that require a logarithmic number of $1$-extendable parts. Moreover, we provided an upper bound of $O(\sqrt{n})$ colors, though closing the gap for the 1-extendable chromatic number remains a challenging task.

After examining arbitrary graphs through modular decomposition, the next logical step is to focus on restricted graph classes commonly used in wireless network contexts. Geometric graphs, such as (unit) disk graphs, present an intriguing class for further study of the 1-extendable chromatic number.

\begin{credits}
\subsubsection{\ackname} We thank Carl Feghali and Pierre Bergé for their valuable discussions on the subject.
\end{credits}

\newpage
\bibliographystyle{plainurl}
\bibliography{biblio.bib}

\end{document}